\documentclass{article}
\usepackage[hidelinks]{hyperref}
\usepackage[preprint]{neurips_2026}

\usepackage{algorithm}
\usepackage{algpseudocode}
\usepackage{amsmath,amssymb,amsfonts,amsthm}
\usepackage{mathtools}
\usepackage{tcolorbox}
\usepackage{algorithm}
\usepackage{algpseudocode}
\usepackage{float}
\usepackage{enumitem}
\usepackage{xcolor}

\usepackage{algpseudocode}

\algnewcommand{\Input}{\item[\textbf{Input:}]}
\algnewcommand{\Output}{\item[\textbf{Output:}]}

\usepackage{booktabs}
\usepackage{xcolor}
\usepackage{enumitem}
\usepackage{subcaption}
\usepackage{url}
\usepackage{graphicx}

\usepackage[T1]{fontenc}

\newtheorem{theorem}{Theorem}[section]
\newtheorem{lemma}[theorem]{Lemma}

\newtheorem{corollary}[theorem]{Corollary}

\newtheorem{remark}[theorem]{Remark}
\newtheorem{hypothesis}[theorem]{Hypothesis}

\newcommand{\norm}[1]{\|#1\|}

\title{Approximate Algorithms for Chamfer Distance Under Translation}

\author{
  \textbf{Gil Halevi} \\
  Columbia University \\
  \texttt{gh2707@columbia.edu}
  \and
  \textbf{Daniel Zhang}\\
  Columbia University \\
  \texttt{dtz2104@columbia.edu}
  \and
  \textbf{Jason Zhang}\\
  Columbia University \\
  \texttt{jz3960@columbia.edu}
}

\newcommand{\eps}{\varepsilon}
\newcommand{\op}{\operatorname}

\begin{document}

\maketitle

\begin{abstract}
Given two sets of points A and B, $|A| = m$, $|B| = n$, the Chamfer distance from $A$ to $B$ is defined as $\operatorname{CD}(A,B) = \sum_{a\in A} \min_{b\in B} d(a,b)$, where $d$ is a distance metric. Chamfer distance is a popular measure of dissimilarity between two sets of points that has seen increasing usage in computer vision and information retrieval as a substitute for the more computationally demanding Earth Mover's distance. We propose a new problem, Chamfer distance under translation, defined as $\operatorname{CDuT}(A,B) :=\min_{t\in \mathbb{R}^d} \operatorname{CD}(A+t,B)$, where $A+t$ denotes the translation of every point in $A$ by $t$. Chamfer distance under translation is valuable in cases where translations capture aspects of the data unlikely to be relevant for dissimilarity, such as temporal, spatial, or other semantic information. For Chamfer distance under translation, we provide four algorithms: (1) an exact quadratic time algorithm in one dimension, (2) a near quadratic time ($2+\varepsilon$)-approximation algorithm in higher dimensions, (3) a $(1+\varepsilon)$-approximation algorithm with running time $\mathcal{O}(mn^2\varepsilon^{-(d+1)})$, and (4) a near-quadratic time  $(1+\varepsilon)$-approximation algorithm for answering the decision version of $\operatorname{CDuT}$ given a separation assumption on $B$. We additionally explore the fine-grained complexity of $\operatorname{CDuT}$.

\end{abstract}

\section{Introduction}

For any two point sets $A, B \subset \mathbb{R}^d$ where $|A| = m$, $|B| = n$, the Chamfer distance from $A$ to $B$ is defined as
\[\operatorname{CD}(A,B) = \sum_{a\in A} \min_{b\in B}d(a,b).\]
where $d$ is some distance metric. Unless otherwise specified, assume $\ell_2$ distance.

$\operatorname{CD}$ has an exact algorithm in quadratic $\mathcal{O}(mn)$ running time: for each $a \in A$, find  the distance to the nearest neighbor $b \in B$ , and sum over all $a \in A$. In many applications, however, the absolute locations of the two point sets are not meaningful. For example, a template shape may appear anywhere in a larger image, a pattern of events may recur at different locations in spatial data, and two collections of deep representations may differ by a global offset while preserving similar relative structure. It is therefore natural to define a translation-invariant version of Chamfer distance, in which one is allowed to translate one point set before measuring its distance to another.

Let the Chamfer distance under translation from $A$ to $B$ be defined as
\[\operatorname{CDuT}(A,B) = \underset{t}{\min} \sum_{a\in A} \min_{b\in B}d(a+t,b).\]
where $d$ is some $\ell_p$ distance metric, such as Euclidean or Manhattan, and $t \in \mathbb{R}^d$ is a translation vector.

Unlike standard Chamfer distance, $\op{CDuT}$ does not have an immediate brute force algorithm: there are uncountably infinite possible translations $t \in \mathbb{R}^d$. Our main algorithmic approach is to replace this continuous search space by finite sets of candidate translations. We show these finite sets are sufficient to obtain an exact algorithm in one dimension and approximation algorithms for higher dimensions.
\subsection{Our Results}
Our contributions are as follows:
\begin{enumerate}[label=(\alph*)]
    \item For one-dimensional point sets, we give an exact algorithm for Chamfer distance under translation that runs in $\mathcal{O}(mn\log(mn))$ time.
    \item In higher dimensions, we also present an algorithm computing a $(2+\eps)$-approximation of $\operatorname{CDuT}$ with high probability in $\mathcal{O}\left(\frac{(m+n)n}{\eps^3}d\log(n)\log(\frac{mn}{\eps})\right)$ time and an algorithm for computing a $(2+\eps)c$ approximation with high probability in $\mathcal{O}\left(\frac{1}{\eps}mn^{1+\frac{1}{2c^2-1}}\log(n)d\right)$ time.
    \item In higher dimensions, we also present an algorithm for computing a $(1+\eps)$-approximation of $\operatorname{CDuT}$ with high probability that has worst-case running time $\mathcal{O}(mn^2\eps^{-(d+1)})$ .
    \item We give a $\tilde{\mathcal{O}}\left(mn^{1+\frac{1}{2c^2-1}}d\right)$ algorithm for answering the approximate decision problem: $\operatorname{CDuT} (A,B)\leq R$ vs. $\op{CDuT}(A,B)>R(1+\eps)$ with high probability. This algorithm relies on a separation assumption: $\forall b_1,b_2\in B, \norm{b_1-b_2}\geq (c+1)(1+\frac{2}{m})R$.
\end{enumerate}

\section{Related Work}
The study of translation-invariant distance measures under translation is well-motivated, and has been studied for a variety of other distance measures including Hausdorff distance under translation (\citet{agarwal-hausdorff-under-translation, fine-grained-hausdorff-under-translation}) and Fréchet distance under translation (\cite{filtser-frechet-distance-under-translation, bringmann-frechet-distance}). Closest to Chamfer distance under translation is Earth Mover's Distance under Translation, which has been studied in (\citet{fine-grained-complexity-emdut}). In this, they prove fine-grained complexity bounds that are equivalent to the fine-grained complexity bounds of Chamfer distance translation, and propose an exact one dimensional algorithm to compute Earth Mover's Distance under Translation in $\tilde{\mathcal{O}}(mn)$ which inspired our exact algorithm. Classical chamfer matching also considers transformations in computer vision, but usually in discretized image/template settings rather than the continuous finite-point-set approximation problem we study.

\section{Algorithms and Analyses}
In this section, we establish our algorithm for computing exact $\operatorname{CDuT}$ for pointsets $A,B \subset \mathbb{R}$ and our algorithms for computing approximate $\operatorname{CDuT}$ for pointsets $A,B \subset \mathbb{R}^d$ for $d > 1$. Additionally, we present an algorithm to solve the approximate decision problem of whether $\operatorname{CDuT}(A,B)\leq R$.

\subsection{Notation}
Let $\op{CD}(A+t,B)$ denote the Chamfer distance achieved by translating every point in $A$ by $t$: $\op{CD}(A+t, B)= \sum_{a\in A} \min_{b\in B}d(a+t,b)$. Let $t^*$ denote any translation that minimizes this Chamfer distance, $t^*\in \arg\min_{t\in \mathbb{R}^d}\op{CD}(A+t,B)$, and let $\mathrm{OPT}:=\operatorname{CD}(A+t^*,B) = \op{CDuT}(A,B)$ denote this minimum distance. We will use $b_{a+t}$ to denote any nearest neighbor to $a$ under translation $t$: $\arg\min_b\norm{a+t-b}$.

\subsection{Algorithm 1: Critical Points Sweeping Exact Chamfer Distance in One Dimension}
Algorithm 1 relies on the following lemma.
\begin{lemma}[Optimality of candidates in one dimension]
    For given pointsets $A, B$, let $t^* \in \mathbb{R}^1$ be an optimal translation for the Chamfer distance under translation, and define $T := \{b_i - a_j : a_i \in A, b_j \in B\}$. Then $\exists t \in T$ such that \[\op{CD}(A + t, B) = \op{CD}(A + t^*, B).\]
    \label{lem:opt-1d}
\end{lemma}
\begin{proof}[Proof sketch]
\begin{remark}
     Let \[T_A := \bigcup_{a \in A} \{b_i - a : i \in [n]\} \cup \left\{ \dfrac{b_{i + 1} + b_i}{2} - a : i \in [n - 1] \right\}.\]
     Let $\operatorname{CDpT}_{A, B} (t) = \operatorname{CD}(A + t, B)$. $\operatorname{CDpT}_{A, B}(t)$ is continuous over $\mathbb{R}^1$, and differentiable for all $t \notin T_A$. For $t \notin T_A$, \[\operatorname{CDpT}_{A, B}'' (t) = 0.\]
     \label{remark:piecelinear}
 \end{remark}

 By Remark \ref{remark:piecelinear}, all points of undifferentiability of $\operatorname{CDpT}_{A, B}$ are in $T_A$. For all the points $t$ in $T_A$ not in $T$ - we call these midpoint translations - for small enough $\eps > 0$, $\operatorname{CDpT}_{A, B}(t + \eps)$ or $\operatorname{CDpT}_{A, B}(t - \eps)$ is guaranteed to be smaller than $\operatorname{CDpT}_{A, B}(t)$, and thus $t$ is not the global minimum. Furthermore, by Remark \ref{remark:piecelinear}, any translation $t_0$ such that $\operatorname{CDpT}_{A, B}' (t_0) = 0$ - except for $t_0$ such that there is guaranteed to be a lower value of $\operatorname{CDpT}$ - must be sandwiched by two points in $T$. Therefore by Fermat's Theorem of Stationary Points, the global minimum must be in $T$. See Appendix \ref{appx:opt-1d} for full proof.
 \end{proof}

By Lemma \ref{lem:opt-1d}, an optimal translation can be found by checking every translation in the set $T$. This set has size $m n$, but checking each naively takes $\mathcal{O}(m^2 n^2)$ time. Instead, this time complexity can be improved by checking the superset $T_A$ in order. The main idea is that $\operatorname{CDpT}_{A, B}$ is piecewise linear, and the change in derivative of $\operatorname{CDpT}$ is easily computable when sweeping across $T_A$. This allows for fast in-order computation of $\operatorname{CDpT}_{A, B}$ for all $T_A$: Formally, let $t_i, t_j$ be points in $T_A$ such that $t_i < t_j$ and $\forall t \in T_A, t \notin (t_i, t_k)$, then $\forall t_k \in (t_i, t_j)$, \[\operatorname{CDpT}_{A, B} (t_j) = \operatorname{CDpT}_{A, B} (t_i) + (t_j - t_i) \cdot \operatorname{CDpT}'_{A, B} (t_k).\]

Evaluation of $\operatorname{CDpT}_{A, B} (t)$ relies solely on the type of translation $t \in T_A$ is.  For $t \in T_A$, for small enough $\eps > 0$, \[\operatorname{CDpT}_{A, B}' (t - \eps) = \operatorname{CDpT}_{A, B}' (t + \eps) - 2 I_\text{match} + 2 I_\text{mid},\]
Where $I_\text{match}$ is the number of pairs $a_j \in A, b_i \in B$ such that $t = b_i - a_j$, and $I_\text{mid}$ is the number of pairs $a_j, b_i$ such that $t = (b_{i + 1} + b_i) / 2 - a_j$. Both $I_\text{match}$ and $I_\text{mid}$ can be counted during the creation of $T_A$, allowing updates from adjacent points in $T_A$ in constant time. Finally, we remark that for $t_l < \min T_A$, $\operatorname{CDpT}_{A, B} '(t_l) = - m$.

The size of $T_A$ is less than $2 m n$, and since each update and shift between elements $T_A$ takes constant time, the time is bounded by finding, sorting, and iterating through $T_A$ and computing the initial Chamfer distance under the smallest translation. The latter takes $\mathcal{O}(mn)$ time. Computing and sorting $T_A$ takes $\mathcal{O}(m n \log (m n))$, and thus the time complexity of the algorithm is $\mathcal{O}(m n \log (m n))$. We show that this is optimal up to logarithmic factors in Section \ref{section:fine-grained-complexity}. The exact pseudo code for algorithm 1 is in Appendix \ref{appx:algo-1-pseudo-code}. Details on how this algorithm extends in higher dimensions for $\ell_1$ and $\ell_\infty$ norm to solve $\operatorname{CDuT}$ in $\mathcal{O}(m^{d + 1} n^{d + 1})$ can be found in Appendix \ref{appx:algo-1-extension}.

\subsection{Preliminaries for the Approximation Algorithms}

Algorithms 2 and 3 rely on the following lemmas.
\begin{lemma}[Upper bound for closest pair under optimal translation] 
\label{lem:closest-pair-upper-bound} Let $t^*$ be an optimal translation, and let \((a^*,b^*)\in A\times B\)
be a closest pair under \(t^*\); $(a^*,b^*) \in \arg\min_{(a,b)\in A \times B}\|a+t^*-b\|$.\\
Then:
\begin{enumerate}
    \item The closest pair satisfies
    \[\|(b^*-a^*)-t^*\| \leq \frac{\mathrm{OPT}}{m}.\]
    \item Furthermore, for every \(\eps \in (0,1]\), at least \(\lfloor m\eps/2\rfloor\) points  \(a\in A\) satisfy
    \[\|(b_{a+t^*}-a)-t^*\| \leq (1+\eps)\frac{\mathrm{OPT}}{m}.\]
\end{enumerate} 
\end{lemma}
\begin{proof}
\begin{align*}
\mathrm{OPT} = \sum_{a\in A} \|a + t^*- b_{a+t^*}\| &\geq m \|a^*+t^*-b^*\|\\
\implies \|(b^* - a^*) -t^* \| &\leq \frac{\mathrm{OPT}}{m}
\end{align*}
For the second claim, sample $a\sim U[A]$ uniformly at random. Then
\[\mathbb{E}_{a\sim U[A]}\big[\norm{(b_{a+t^*}-a)-t^*} \big]=\frac{1}{m}\sum_{a \in A} \norm{a+t^*-b_{a+t^*}} =\frac{\mathrm{OPT}}{m}.\]
By Markov's inequality, for $\eps\in (0,1]$, we have
\[\Pr_{a\sim U[A]}\left[\norm{(b_{a+t^*}-a)-t^*} \leq (1+\eps)\frac{\mathrm{OPT}}{m}\right]\geq 1- \frac{1}{1+\eps} = \frac{\eps}{1+\eps}\geq \frac{\eps}{2}.\]
Thus, at least $\lfloor m\frac{\eps}{2}\rfloor$ points $a\in A$ satisfy $\|(b_{a+t^*}-a)-t^*\| \leq (1+\eps)\frac{\mathrm{OPT}}{m}$
\end{proof}
Importantly, this shows that the finite candidate set \(\{b_i-a_j : a_j\in A,\ b_i\in B\}\) of size $mn$ contains a translation close to an optimal translation. Moreover, for \(\eps\in(0,1]\), there are at least \(\lfloor m\eps/2\rfloor\) choices of \(a\in A\) with a candidate translation \(b-a\) within distance \((1+\eps)\mathrm{OPT}/m\) of \(t^*\).
\begin{lemma}[Approximation factor of Chamfer distance under candidate translations]
\quad
\begin{enumerate}
\item For every translation $t$ with $\norm{t-t^*}\leq \frac{k}{m}\mathrm{OPT}$ to some $t^*$, \[\op{CD}(A+t,B)\leq (1+k)\mathrm{OPT}.\]
\item \[\mathrm{OPT} \leq \min_{a\in A,b\in B}\operatorname{CD}(A+b-a, B)\leq 2\mathrm{OPT}.\]
\item For at least $\lfloor\ m\eps/2\rfloor$ points $a\in A$, \[\mathrm{OPT}\leq \min_b\operatorname{CD}(A+b-a,B)\leq (2+\eps)\mathrm{OPT}.\]
\end{enumerate}
\label{lem:appx-chamfer-translations}

\end{lemma}
\begin{proof}[Proof sketch]
In the worst case, $\mathrm{CD}(A+t,B)$ involves moving every $a$ to its position under $t^*$, and then moving it to $b_{a+t^*}$. By Lemma \ref{lem:closest-pair-upper-bound}, there is at least one candidate translation that achieves $\norm{t_c-t^*}\leq \frac{\mathrm{OPT}}{m}$, and at least $\lfloor \frac{\eps m}{2} \rfloor$ that achieve $\norm{t_c-t^*}\leq (1+\eps)\frac{\mathrm{OPT}}{m}$, and thus the minimal chamfer distance at any candidate translation is upper bounded by $2\mathrm{OPT}$ and $(2+\eps)\mathrm{OPT}$ respectively.
\end{proof}

\subsection{Algorithm 2: Near-Quadratic Time for Computing 2-Approximate Solution}
In this section we present an $\mathcal{O}(\frac{(m+n)n}{\eps^3}d\log(n)\log(\frac{mn}{\eps}))$ time algorithm for computing a $(2+\eps)$- approximation of $\operatorname{CDuT}$ (variant 1), and an $\mathcal{O}(\frac{1}{\eps}mn^{1+\frac{1}{2c^2-1}}d\log(n))$ time algorithm for computing a $(2+\eps)c$-approximation of $\operatorname{CDuT}$ (variant 2), for $c>1$ and $\eps\in (0,1)$, both with high probability. 

The basic approach is to sample $\mathcal{O}(\frac{1}{\eps})$ points $a_1,...,a_k$ uniformly at random from $A$. For each of them, for each $b\in B$, compute an approximation for $\operatorname{CD}(A+b-a_i,B)$, and return the minimum such approximation found.

There are two different approaches for computing $\operatorname{CD}(A+b-a_i,B)$. One of them, which we'll call variant 1, is to call the algorithm from \citet{near-linear-time-algo-chamfer-distance} for a $1+\frac{\eps}{4}$ approximation of $\operatorname{CD}$ in time $\mathcal{O}(\frac{1}{\eps^2}n\log(n)d)$ as a subroutine with probability $\geq 0.99$, which can be boosted to $1-\frac{0.05}{kn}$ by repeating a logarithmic number of times and taking the median. This leads to total time $\mathcal{O}(\frac{(m+n)n}{\eps^3}d\log(n)\log(\frac{mn}{\eps}))$ for a $(2+\eps)$-approximation with probability $\geq 0.9$, by running it a logarithmic number of times.

Since in many contexts $n \gg m$, we present an alternative method, variant 2, based on approximate nearest neighbor in order to reduce the exponent on $n$. As described in further detail in Appendix \ref{appx:nn-data-structure}, during preprocessing we build a multi-scale LSH data structure on $B$ in space $\mathcal{O}(n^{1+\frac{1}{2c^2-1}}\log(n))$ and time $\mathcal{O}(n^{1+\frac{1}{2c^2-1}}\log(n))$ for giving $c$-approximations to nearest neighbor queries in time $\mathcal{O}(n^{\frac{1}{2c^2-1}}\log(n)d)$. For each translation $t$, for each point $a_i\in A$, we query the data structure with vector $a_i+t$ and store the approximate distance to its nearest neighbor $\Delta_i= \norm{\hat{b}_{a_i+t}-a_i}$. We use the $\sum_{i\in [m]}\Delta_i$ as an approximation of the Chamfer distance at $t$, and return the minimum such sum over all $\mathcal{O}(\frac{1}{\eps})$ translations.

\subsection{Algorithm 3: Sampled Local-Net $(1+\eps)$-Approximation for $\operatorname{CDuT}$}
\begin{figure}[t]
\centering
\begin{tcolorbox}[title=\textsc{$(1+\eps)$-ChamferDistanceUnderTranslation}]

\begin{algorithmic}
\State Let \(Q\) be a uniformly sampled random subset of \(A\) of size \(
k=\left\lceil \frac{2}{\eps}\log\frac{1}{\delta}\right\rceil \)
\State Let \(T_{samp} =  \{b_i-a_j : a_j\in Q,\ b_i\in B\}\)
\State Let \(u = \min_{t_c\in T_{samp}}\operatorname{CD}(A+t_c,B)\)
\State Set \(R = (1+\eps)\frac{u}{m}\) and \(\rho = \frac{\eps u}{3m}\)\\
\State For each $t_c \in T_{samp}$:
\State \qquad Construct a local $\rho$-net $N_{t_c}$ of the ball $B(t_c,R)$
\State \qquad Compute $\operatorname{CD}(A+\hat t,B)$ for every $\hat t \in N_{t_c}$

\State Let \[N := \bigcup_{t_c \in T_{samp}} N_{t_c}\]
\State Return \[
\hat t = \arg\min_{s \in N} \operatorname{CD}(A+ s, B),
\widehat{\mathrm{OPT}} =  \operatorname{CD}(A+ \hat t, B)
\]

\end{algorithmic}

\end{tcolorbox}
\caption{The $(1+\eps)$-ChamferDistanceUnderTranslation Algorithm}
\label{fig:chamfer-1+e-appx}
\end{figure}
In this section we present an algorithm for computing a $(1+\eps)$- approximation for $\op{CDuT}$. Pseudocode for the algorithm can be seen in Figure \ref{fig:chamfer-1+e-appx}.

The algorithm is based on Lemma~\ref{lem:closest-pair-upper-bound}, which shows that the finite set of candidate translations contains points close to an optimal $t^*$. We therefore search local $\rho$-nets of balls centered at these candidate translations. Taking the nets sufficiently fine guarantees that at least one net point lies within $\mathcal{O}(\frac{\eps\op{OPT}}{m})$ of $t^*$, and this proximity is enough to ensure that the resulting Chamfer cost is at most $(1+\eps)\op{OPT}$.
\begin{theorem}[Sampled Local-Net $(1+\eps)$-Approximation for $\operatorname{CDuT}$]
\label{thm:local-net-algo-theorem}
Local-Net-$\operatorname{CDuT}$-Estimate runs in time $\mathcal{O}(mn^2\eps^{-(d+1)})$ and returns a translation and value estimate $(\hat t,\widehat{\op{OPT}})$ such that with high probability,
\[\op{CDuT}(A,B) \leq \widehat{\op{OPT}} \leq (1 + \eps)\op{CDuT}(A,B),\]

when the underlying metric is any $\ell_p$ metric.

\end{theorem}

We state Algorithm 3 with exact candidate evaluation. Near-linear randomized estimators for fixed-translation Chamfer distance could be substituted \citet{near-linear-time-algo-chamfer-distance}, but their high-probability guarantees would need to hold uniformly over all net candidates, introducing additional amplification and variance-dependent factors.
\begin{proof}[Proof sketch]
Let \(t^*\) be an optimal translation and let $\mathrm{OPT}=\op{CDuT}(A,B)$. By Lemma~\ref{lem:closest-pair-upper-bound}, there are many candidate translations \(b-a\) within distance \((1+\gamma)\mathrm{OPT}/m\) of \(t^*\). Therefore, if \(Q\subseteq A\) is sampled uniformly without replacement with \(|Q|=\lceil (2/\gamma)\log(1/\delta)\rceil\), then with probability at least \(1-\delta\), the sampled candidate set \(T_{\mathrm{samp}}=\{b-a:a\in Q,\ b\in B\}\) contains some \(t_c\) satisfying \(\|t_c-t^*\|\le (1+\gamma)\mathrm{OPT}/m\). Let \(u=\min_{t_c\in T_{\mathrm{samp}}}\op{CD}(A+t_c,B)\). On this event, Lemma \ref{lem:appx-chamfer-translations} gives \(\mathrm{OPT}\le u\le (2+\gamma)\mathrm{OPT}\).

Set \(R=(1+\gamma)u/m\) and \(\rho=\eps u/(hm)\) for \(h\ge 2+\gamma\). The successful candidate satisfies \(\|t_c-t^*\|\le (1+\gamma)\mathrm{OPT}/m\le R\), so the ball \(B(t_c,R)\) contains \(t^*\), and its \(\rho\)-net contains some \(s\) with \(\|s-t^*\|\le \rho\). Since \(\|s-t^*\|\le\rho\), each nearest-neighbor distance in the Chamfer sum increases by at most \(\rho\), so \(\op{CD}(A+s,B)\le \op{CD}(A+t^*,B)+m\rho=\mathrm{OPT}+m\rho=\mathrm{OPT}+\eps u/h\le (1+\eps)\mathrm{OPT}\), so minimizing over all net points gives the desired approximation. Finally, \(|T_{\mathrm{samp}}|=\mathcal{O}((n/\gamma)\log(1/\delta))\), each local net has size \((R/\rho)^d=\mathcal{O}(\eps^{-d})\), and each exact Chamfer evaluation costs \(\mathcal{O}(mn)\), giving total time \(\mathcal{O}(mn^2\gamma^{-1}\eps^{-d}\log(1/\delta))\), which becomes \(\mathcal{O}(mn^2\eps^{-(d+1)})\) for \(\gamma=\eps\) and constant \(\delta\). Refer to Appendix \ref{appx:proof-of-local-net-theorem} for the full proof.
\end{proof}

Note that if the candidate translations are well clustered in translation space, the dependence on the number of sampled candidates $|T_{samp}|$ can be significantly reduced. For example, suppose that the sampled translations lie in a region whose $R$-covering number is $K$. Then the number of net points needed for the union of the search balls is on the order of $K\eps^{-d}$ rather than $|T_{samp}|\eps^{-d}$. This improves the bound when $K\ll|T_{samp}|$. This is precisely the regime in which $A$ and $
B$ are similar up to translation: many sampled offsets $b-a$ concentrate in the same region of translation space, causing their local search balls to overlap. More discussion can be found in Appendix \ref{appx:algo-3-instance-sensitive-refinement}.

The exponential dependence on $d$ is consistent with our fine-grained lower bound. In Section \ref{section:fine-grained-complexity}, we show that, assuming ETH, there is no $(\frac{n}{\eps})^{o(d)}$-time $(1+\eps)$-approximation algorithm for $\op{CDuT}(A,B)$. Thus, substantially improving the $\eps^{-d}$ dependence in Algorithm 3 would require breaking standard complexity assumptions.

\subsection{Algorithm 4: Answering Decision Problem up to $(1+\eps)$ Given Separation Assumption}
Algorithm 4 is specifically for the decision version of the problem: Given a set of points $A,B$, is $\op{CDuT}(A,B)\leq R$? It returns YES if $\op{CDuT}(A,B)\leq R$, NO if $\op{CDuT}(A,B)> R(1+\eps)$, and either if $ R< \op{CDuT}(A,B) \leq R(1+\eps)$.

It relies on an assumption: $\forall b_1,b_2 \in B, \norm{b_1-b_2}\geq (c+1)(1+\frac{2}{m})R$. We note that this is a strong assumption, though not entirely impractical. One can think of Algorithm 4 as a form of soft fingerprinting, designed for telling whether one dataset is found almost identically in another after some translation; the interpoint distances in the bigger dataset should dwarf the desired test threshold. Given these assumptions, we present an $\mathcal{O}(mnd(n^{\frac{1}{2c^2-1}}\log(n)+\log^{3}(\frac{mR}{\eps})))=\tilde{\mathcal{O}}(mn^{1+\frac{1}{2c^2-1}}d)$ time algorithm for solving this decision version of Chamfer distance under translation. Note that $c$ here just changes the point separation threshold, not the approximation factor of the algorithm.

The algorithm is very similar to variant 2 of Algorithm 2. First, we sample $\mathcal{O}(1)$ points $\tilde{a}_{1,...k}\in A$. For every such $\tilde{a}_i$, for every $\tilde{b}_j\in B$, consider the translation $\tilde{t}=\tilde{b}_j-\tilde{a}_i$, as in algorithm 2. At every such translation $\tilde{t}$, we compute the set of difference vectors under its closest point assignment.
\[D_{\tilde{t}}=\{\vec{\Delta}_i|i\in [m]\},\quad \vec{\Delta}_i=b_{a_i+\tilde{t}}-a_i \quad \forall i\in [m].\]
The nearest neighbors $b_{a_i+\tilde{t}}$ are found by approximate nearest neighbor. We then take all these difference vectors and compute an approximation $\hat{\mu}_1$ of their geometric median $\mu_1\in \arg\min_{t\in \mathbb{R}^d} \sum_i\norm{\vec{\Delta_i}-t}$ as well as the total distance to it $s_{\hat{\mu}_1}(D_{\tilde{t}})=\sum_{i\in [m]}\norm{\vec{\Delta}_i-\hat{\mu}_1}$. This geometric median is often in practice found using Weiszfeld's algorithm (\citet{somme-des-distance}), but there is also a provably efficient approximation algorithm by Cohen et al. \citet{cohen-geometric-mean} that takes time $\mathcal{O}(md\log^3(\frac{mR}{\eps}))$ and returns a point $\hat{\mu}_1$ that is $\frac{R \eps}{m}$-close to $\mu_1$. We return YES if one of these approximate medians yields $ s_{\hat{\mu}_1}\leq R(1+\eps)$, and NO otherwise.

This algorithm's correctness relies on the fact that under the separation assumption and $\mathrm{OPT}\leq R$, the nearest neighbors under one $\tilde{t}$ found by approximate nearest neighbor will be exactly the nearest neighbors under $t^*$. For this, we introduce the concept of alignment and $c$-uniqueness.
\begin{lemma}[Alignment of close points and c uniqueness under separation assumption]
Suppose $\mathrm{OPT}\leq R$ and that $
\forall b_1,b_2\in B,\norm{b_1-b_2}\geq (c+1)(1+\frac{2}{m})R$. Call a translation $t$ aligned with $t'$ if $\forall i \in [m],b_{a_i+t}=b_{a_i+t'}$. Call $t$ $c$-unique if $b_{a_i+t}$ is the nearest neighbor to $a_i+t$ by a factor of c, ie $\forall b'\neq b_{a_i+t}\in B,\norm{a_i+t-b'}\geq c\norm{a_i+t-b_{a_i+t}}$. Then:
\begin{center}
$\norm{t-t^*}\leq \frac{2R}{m} \implies t$ is aligned with $t^*$ and $c$-unique.
\end{center}
\begin{proof}[Proof sketch]
Since any two $b\in B$ are separated by greater than $2R$, and since the distance $\norm{b_{a_i+t^*}-a_i}\leq R$, translating each $a_i$ by a small amount is unable to change which $b\in B$ is closest to $a_i$, which proves alignment. $c$-uniqueness is the natural extension when the distances are greater than $(1+c)R$.
\end{proof}

\label{lem:c-alignment-under-separation}
\end{lemma}
    
We now want some way of relating the total distance of a set of differences to its geometric median to Chamfer distance. For this we introduce the following lemmas.

\begin{lemma}[Correspondence between total distance and Chamfer distance]
For a set of points $P=\{p_1,...,p_k\}\in \mathbb{R}^d$, define their total distance to a point p as $s_{p}(P)=\sum_i \norm{p_i-p}$. Consider the set of differences induced by a point assignment from $A$ to $B$, $D=\{b_{i}-a_i|i=1,...,m\}$ for some $b_1,...,b_m\subset B$. Then,
\begin{enumerate}
    \item \[\forall t\in \mathbb{R}^d, s_{t}(D)\geq \op{CD}(A+t,B).\]
    \item 
Define the set of differences induced by a closest point assignment under some translation $t$ as $D_t=\{b_{a_i+t}-a_i|i=1,...,m\}$. Then,
\[s_t(D_t)=\op{CD}(A+t,B).\]
\end{enumerate}
\begin{proof}
\[s_t(D)=\sum_{i\in [m]} \norm{(b_i-a_i)-t}=\sum_{i\in [m]} \norm{a_i+t-b_i}\geq \sum_{i\in [m]} \min_b\norm{a_i+t-b}=\operatorname{CD}(A+t,B)\]
When $D=D_t$, the inequality becomes equality.
\end{proof}
\label{lem:total-distance-CD}
\end{lemma}
\begin{corollary}
    If $t'$ is aligned with $t$, $s_{t'}(D_t)=\op{CD}(A+t',B)$
\label{cor:alignment-CD}
\end{corollary}
\begin{lemma}

Define a geometric median of a set $P$ as a point (not necessarily in the set) that minimizes the total distance $\mu_1 \in \arg\min_{p\in \mathbb{R}^d}s_p(P)$. If $t$ is aligned with some $t^*$, then every geometric median of $D_t$ has $\operatorname{CD}(A+\mu_1,B)=\mathrm{OPT}$; put differently, every $\mu_1$ is a valid $t^*$.
\begin{proof}
\[s_{t^*}(D_t)=\sum_{i\in [m]} \norm{(b_{a_i+t}-a_i)-t^*}=\sum_{i\in [m]} \norm{a_i+t^*-b_{a_i+t^*}}=\mathrm{OPT}\]
Since $\forall t', s_{t'}(D_t)\geq \operatorname{CD}(A+t',B)$, the lowest value that $s_{\mu_1}(D_t)$ can possibly achieve is $\mathrm{OPT}$, and $s_{t'}(D_t)=\mathrm{OPT} \implies t'$ is a valid $t^*$.
\end{proof}
\label{lem:median-opt}
\end{lemma}
Finally, this allows us to prove the correctness of algorithm 4.
\begin{theorem}[Correctness of algorithm 4]
Under the assumption that $\forall b_1,b_2 \in B, \norm{b_1-b_2}\geq (1+c)(1+\frac{2}{m})R$, algorithm 4 will return YES if $\mathrm{OPT}\leq R$, and NO if $\mathrm{OPT}>R(1+\eps)$
\end{theorem}
\begin{proof}
If $\mathrm{OPT}\leq R$, then with high probability some $\tilde{t}$ will be $\frac{2R}{m}$-close to some $t^*$ and thus aligned with it and $c$-unique. For this $\tilde{t}$, by Lemma \ref{lem:median-opt} every $\mu_1$ achieves value $s_{\mu_1}(D_{\tilde{t}})=\mathrm{OPT}$, and since $\tilde{t}$ is $c$-unique we will find its exact nearest neighbors by approximate nearest neighbor. $\hat{\mu}_1$ will then be $\frac{\eps R}{m}$-close to some $\mu_1$, thus $s_{\hat{\mu}_1}(D_{\tilde{t}})\leq s_{\mu_1}(D_{\tilde{t}})+m\norm{\hat{\mu}_1-\mu_1}\leq \mathrm{OPT}+ \eps R\leq (1+\eps)R$. The algorithm will therefore return YES.

If $\mathrm{OPT}> (1+\eps)R$, we have $\forall \hat{\mu}_1,s_{\hat{\mu}_1}(D_t)\geq \mathrm{CD}(A+\hat{\mu}_1,B)> (1+\eps)R$ and so the algorithm will return NO.
\end{proof}
Given a further separation assumption on $A$, $\forall a_1,a_2\in A$, $\norm{a_1-a_2}>(1+\eps)R$, one can show that the nearest neighbors $b_{a_i+t^*}$ to all $a_i$ under $t^*$ are all unique. This means that Algorithm 4 will be valid for Earth Mover's Distance under translation as well. Given the separation assumption on $B$, this is a natural assumption to make on $A$ as well. See Appendix \ref{appx:further-assumption-emd} for formal proof.

\section{Fine-Grained Complexity}
\label{section:fine-grained-complexity}
The Orthogonal Vectors problem is: given vector sets $A, B \subseteq \{0, 1\}^d$ and $|A| = |B| = n$, return whether there exists $x \in A, y \in B$ such that $x \cdot y = 0$. Naively, this can be solved in $\mathcal{O}(n^2 \cdot d)$ time by iterating through all possible pairs.

\begin{hypothesis}[Orthogonal Vectors Hypothesis]
    No algorithm solves the Orthogonal Vectors problem in $\mathcal{O}(n^{2 - \delta})$ for $\delta > 0$.
\end{hypothesis}

Strong Exponential Time Hypothesis (SETH) implies the Orthogonal Vector Hypothesis, and furthermore, reductions to the Orthogonal Vectors problem is a commonly used tool in proving theoretical lower bounds in Fine-Grained Complexity: including proving that assuming OVH is true, that Hausdorff Distance under Translation and Earth Mover's Distance under Translation cannot be computed in strongly subquadratic time (\cite{fine-grained-hausdorff-under-translation, fine-grained-complexity-emdut}). A similar theoretical bound can be proved for Chamfer Distance under translation:

\begin{theorem}[Lower Bound of $\operatorname{CDuT}$ in 1D]
    Assuming OVH, there is no algorithm that given sets $A, B \subseteq \mathbb{R}^1$ and $|A| \leq n, |B| \leq n$ computes $\operatorname{CDuT}(A, B)$ in time $\mathcal{O}(n^{2 - \delta})$ for $\delta > 0$.
    \label{thm:lower-bound-1d}
\end{theorem}

The proof of this theorem follows almost without modification from the proof of conditional lower bound of Earth Mover's Distance under Translation, found in \citet{fine-grained-complexity-emdut}. We define vector gadget construction for a reduction from OV to $\operatorname{CDuT}$:

For a vector $x \in A \subseteq \{0, 1\}^d$, define the vector gadget mapping $A(x)$ into a set of points as:
\begin{itemize}
    \item add one point at 0 and one point at $4 d + 1$,
    \item for every $i \in \{1, ..., d\}$:
    \begin{itemize}
        \item if $x_i = 0$, add points at $\{4 i - 2, 4 i - 1\}$
        \item if $x_i = 1$, add points at $\{4 i - 3, 4 i\}$
    \end{itemize}
\end{itemize}

For a vector $y \in B \subseteq \{0, 1\}^d$, define the vector gadget mapping $B(y)$ into a set of points as:
\begin{itemize}
    \item add one point at 0 and one point at $4 d + 1$,
    \item for every $i \in \{1, ..., d\}$:
    \begin{itemize}
        \item if $y_i = 0$, add points at $\{4 i - 3, 4 i - 2, 4 i - 1, 4 i\}$
        \item if $y_i = 1$, add points at $\{4 i - 2, 4 i - 1\}$
    \end{itemize}
\end{itemize}

With these definitions, we now have this lemma:
\begin{lemma}
    \label{lem:convert-1d-reduction}
    Let $w = 4 d + 1$ be the width of the gadget. Let $x, y \in \{0, 1\}^d$:
    \begin{enumerate}
        \item if $x$ and $y$ are not orthogonal, then $\operatorname{CDuT}(A(x), B(y)) = 0$.
        \item if $x$ and $y$ are not orthogonal, then $\operatorname{CD}(A(x) + t, B(y)) \geq \max\{1, t\}$.
        \item if $|t| \geq w$, then $\operatorname{CD}(A(x) + t, B(y)) = |t| \cdot 2 (d + 1) - 4 d^2 - 5 d - 1$.
    \end{enumerate}
\end{lemma}

If $x$ and $y$ are orthogonal, it is clear that $A(x) \subseteq B(y)$. If they are not orthogonal, the index $i$ in which $x_i = y_i = 1$ creates a chamfer distance of two under no translation, and under translation the most extreme points of $A$ incur distance at least $|t|$. Gadget $A(x)$ also has the same number of points and center mass irrespective of $x$, and therefore with $|t| > w$, the distance does not depend on $x$ or $y$, as the nearest neighbor for all of $A(x)$ will map to one of the extreme points in $B(y)$. With this lemma, the rest of the argument of Theorem 4.2 and OV reduction in \citet{fine-grained-complexity-emdut} holds without modification with $\operatorname{CDuT}$, proving Theorem \ref{thm:lower-bound-1d}.

The $k$-clique problem is given a graph $G = (V, E)$ return whether there is a subgraph of size $k$ that is completely connected. It is known that the Exponential Time Hypothesis implies that $k$-clique cannot be solved in $f(k) \cdot N^{o(k)}$ for any computable function $f$.

\begin{theorem}
    Assuming ETH, there is no algorithm that given sets $A, B \subseteq \mathbb{R}^d$ and $|A| \leq n, |B| \leq n$ computes a $(1 + \eps)$-approximation of $\operatorname{CDuT}(A, B)$ in time $(\frac{n}{\eps})^{o(d)}$.
    \label{thm:lower-bound-high-d}
\end{theorem}
This follows from the $k$-Clique reduction used for EMDuT, in which the constructed instance has ambient dimension \(d=\Theta(k)\).
\begin{lemma}[Gadget Combination Lemma]
    \label{lem:gadget-combination}
    Let $1 \leq p \leq \infty$. Given sets $A_1, B_1, ..., A_k, B_k \subset \mathbb{R}^d$ of total size $n$, in time $\mathcal{O}(n d)$ we can compute $A, B \subset \mathbb{R}^d$ of total size $n$ such that \[\operatorname{CDuT}(A, B) = \min_{t \in \mathbb{R}^d} \sum_{i = 1}^k \operatorname{CD}(A_i + t, B).\]
\end{lemma}
This lemma is true because $A_1, .., A_k$ gadget can be placed in order on a line in $\mathbb{R}^d$ with space $10 n \Delta$ between each adjacent pair in $A$, and $B_1, ..., B_k$ similarly placed in $B$, where $\Delta$ is the diameter of $\bigcup_{i = 1}^k A_i \cup B_i$; with such a large gaps between gadgets, the optimal translation is guaranteed to map all the points in $A_i$ to $B_i$ or suffer a penalty of $4 n \Delta$ which is guaranteed to be more than $\operatorname{CDuT}(A, B)$. Furthermore, the asymmetric reductions from $k$-clique to Earth Mover's Distance under Translation in \cite{fine-grained-complexity-emdut} uses $A$ gadgets of size one, for both $\ell_1$ and $\ell_\infty$. As a result, the necessity of injectivity in mappings between points in the $A$ gadget and points in the $B$ gadget is unused, and the reductions to Earth Mover's Distance under Translation work as reductions to $\operatorname{CDuT}$.

The asymmetric reductions for $\operatorname{CDuT}$ can be modified to a symmetric reduction by expanding each $A$ and $A'$ gadget into $|E|$ points tightly centered around the original point. The lack of the injectivity requirements will result in all the points mapping to the same point, with the effect of this modification resulting in scaling the $\operatorname{CDuT}$ by $|E|$.

\section{Discussion}
Overall, our results indicate that computing Chamfer distance under translation in high dimensions is difficult. Across our algorithms, candidate translations of the form \(b-a\) play a central role: in one dimension, they contain an exact optimum and yield a near-quadratic time algorithm; in higher dimensions, they provide efficient coarse approximations and serve as anchors for stronger guarantees. However, near-optimal approximation in the worst case appears to require additional search around these candidates, leading to exponential dependence on the dimension. Our decision algorithm shows that stronger guarantees are possible under additional geometric structure, but this comes at the cost of a separation assumption. This apparent barrier to dimension-independent near-optimal approximation is not merely an artifact of our algorithm: our fine-grained lower bounds suggest that substantially improving the worst-case dimension dependence would require violating standard complexity assumptions.

Notably, exact computation in higher dimensions remains open for \(\ell_p\) norms with \(1<p<\infty\), while our known exact higher-dimensional extension applies only to the locally linear cases \(\ell_1\) and \(\ell_\infty\).

More broadly, since Chamfer distance is already a relaxation of Earth Mover's Distance, our results suggest that further relaxations may be necessary for scalable translation-invariant comparison in high-dimensional settings.

\newpage

\bibliographystyle{plainnat}
\bibliography{refs}

\newpage

\appendix

\section{Proof of Lemma \ref{lem:opt-1d}}
\label{appx:opt-1d}
\begin{proof}[Proof of Lemma \ref{lem:opt-1d}]
 First, for a given set $B \subset \mathbb{R}$, let function $f_a (t) = \min_{b \in B} \|a - b\|$: this function calculates the minimum distance of $a$ to any point in $B$. Then let $b_0, b_1, ..., b_n \in \mathbb{R}^1$, such that $b_0 < b_1 < ... <b_n$ and $B = \{b_0, b_1, ..., b_n\}$.

 \begin{remark}
     Let \[T_a := \{b_i - a : i \in [n]\} \cup \left\{ \dfrac{b_{i + 1} + b_i}{2} - a : i \in [n - 1] \right\}.\]
     $f_a (t)$ is continuous over $\mathbb{R}^1$, and differentiable everywhere except $T_a$. Furthermore, for some $k \in \mathbb{R}^1$, for all $t \notin T_a$, \[f_a'' (t) = 0.\]
 \end{remark}

 Let the Chamfer distance under translation, parametrized by translation,  be called $\operatorname{CDpT}$:
 \[\operatorname{CDpT}_{A, B}(t) = \op{CD}(A+t,B) = \sum_{a \in A} f_a(t).\]
 The minimum over all translations in $\mathbb{R}^1$ of $\operatorname{CDpT}$ is equal to $\operatorname{CDuT}$.

 \begin{remark}
     Let $T_A := \bigcup_{a \in A} T_a$. $\operatorname{CDpT}_{A, B}(t)$ is continuous over $\mathbb{R}^1$, and differentiable for all $t \notin T_a$. Furthermore, for $t \notin T_A$, \[\operatorname{CDpT}_{A, B}'' (t) = 0.\]
 \end{remark}

 These properties result from $\operatorname{CDpT}_{A, B} (t)$ being the summation of functions that are continuous, differentiable for $t \notin T_A$, and have zero second order derivatives for $t \notin T_A$.

 Additionally, for all $a \in A$, $f_a$ is non-negative, and as a result $\operatorname{CDpT}_{A, B}(t)$ is non-negative. Moreover, $\operatorname{CDpT}_{A, B}(t)\to \infty$ as $|t|\to\infty$, since every point of $A+t$ moves arbitrarily far from the finite set $B$. Therefore $\operatorname{CDpT}_{A, B}$ attains a global minimum. At any global minimizer, either $\operatorname{CDpT}_{A, B}'(t)=0$ or $\operatorname{CDpT}_{A, B}'(t)$ does not exist. These conditions give exactly the set $T_A$.

 Assume $A \neq \emptyset$, and as a result $\operatorname{CDpT}_{A, B}(t)$ is not identically zero. Thus the derivative of $\operatorname{CDpT}_{A, B}(t)$ is not zero over its domain. As a consequence, if there exists $t_i \in \mathbb{R}^1$ such that $\operatorname{CDpT}_{A, B}' (t_i) = 0$, then there must exist $t_j \in \mathbb{R}^1$ such that $\operatorname{CDpT}_{A, B}' (t_j) \neq 0$; without loss of generality, assume $t_i < t_j$. Because for all $t \in \mathbb{R}^1$ $\operatorname{CDpT}_{A, B}'' (t) = 0$ or does not exist, there must exists a $t_k, t_i < t_k < t_j$ such that $\operatorname{CDpT}_{A, B}' (t_k)$ does not exist, and every point in $t_0 \in [t_i, t_k)$ satisfies $\operatorname{CDpT}_{A, B}' (t_0) = 0$.

 Thus, $\operatorname{CDpT}_{A, B} (t_i) = \operatorname{CDpT}_{A, B} (t_k)$. As a result, if the minimum of $\operatorname{CDpT}_{A, B}$ over $\mathbb{R}^1$ is at a $t^*$ where $\operatorname{CDpT}_{A, B}' (t^*) = 0$, there exists $t_k$ such that $\operatorname{CDpT}_{A, B}' (t_k)$ does not exist and $\operatorname{CDpT}_{A, B} (t_i) = \operatorname{CDpT}_{A, B} (t_k)$. Therefore $t_k \in T_A$. Otherwise, the minimum of $\operatorname{CDpT}_{A, B}$ is at $t^*$ such that $\operatorname{CDpT}_{A, B}' (t^*)$ does not exist.

 To complete this proof, we argue that any $t \in T_A, t \notin T$ cannot be the global minimum. Suppose for contradiction this is true. For such $t$, for small enough $\eps > 0$,
 \begin{equation}
     \operatorname{CDpT}_{A, B}' (t - \eps) >  \operatorname{CDpT}_{A, B}' (t + \eps).
     \label{equ:mid-maxness}
 \end{equation}
 Setting the left side of \ref{equ:mid-maxness} to zero results in the derivative of $\operatorname{CDpT}_{A, B}$ for a neighborhood greater than $t$ to be negative, making the value of $\operatorname{CDpT}_{A, B}$ in that neighborhood less than $\operatorname{CDpT}_{A, B}(t)$, contradicting our claim. Setting the right side to zero results in the derivative of $\operatorname{CDpT}_{A, B}$ for a neighborhood less than $t$ to be positive, making the value of $\operatorname{CDpT}_{A, B}$ in that neighborhood less than $\operatorname{CDpT}_{A, B}(t)$, contradicting our claim. Furthermore, the left side of \ref{equ:mid-maxness} cannot be negative while and the right side positive. Therefore $t$ cannot be a local minimum, implying it cannot be a global minimum.
\end{proof}

\section{Proof of Lemma \ref{lem:appx-chamfer-translations}}
\begin{proof}[Proof of Lemma \ref{lem:appx-chamfer-translations}]
Let $\norm{t-t^*}\leq \frac{k}{m}\mathrm{OPT}$. Then
\begin{align*}
\operatorname{CD}(A+ t,B)&=\sum_{a\in A}\min _{b\in B}\norm{a+t-b}\\
&\leq \sum_{a\in A}\norm{a+t-b_{a+t^*}}\\
&\leq \sum_{a\in A}\norm{a+t^*-b_{a+t^*}}+\norm{t-t^*}\\
&\leq \mathrm{OPT}+m\frac{k\mathrm{OPT}}{m}\\
&=(1+k)\mathrm{OPT}.
\end{align*}
Since by Lemma~\ref{lem:closest-pair-upper-bound}, $\norm{(b^*-a^*)-t^*}\leq \frac{\mathrm{OPT}}{m}$, and since $a^*+(b^*-a^*)=b^*$ contributes $0$ to the Chamfer sum,
    \begin{align*}
        \mathrm{OPT}
        &\leq \min_{a\in A,b\in B}\operatorname{CD}(A+(b-a),B)\\
        &\leq \operatorname{CD}(A+ (b^*-a^*),B)\\
        &\leq \mathrm{OPT}-\norm{a^*+t^*-b^*}+(m-1)\norm{(b^*-a^*)-t^*}\\
        &= \mathrm{OPT}+(m-2)\norm{(b^*-a^*)-t^*}\\
        &\leq \left(2-\frac{2}{m}\right)\mathrm{OPT}
        \leq 2\mathrm{OPT}.
    \end{align*}
    Similarly, for at least $\lfloor\eps m/2\rfloor$ points $a_i$ in $A$ (Lemma~\ref{lem:closest-pair-upper-bound}) which have $\norm{(b_{a_i+t^*}-a_i)-t^*}\leq (1+\eps)\frac{\mathrm{OPT}}{m}$,
    \begin{align*}
        \mathrm{OPT}
        &\leq \min_{b\in B}\operatorname{CD}(A+(b-a_i),B)\\
        &\leq \operatorname{CD}(A+ (b_{a_i+t^*}-a_i),B)\\
        &\leq (2+\eps)\mathrm{OPT}.
    \end{align*}
\end{proof}

\section{Proof of Lemma \ref{lem:c-alignment-under-separation}}
\begin{proof}[Proof of Lemma \ref{lem:c-alignment-under-separation}]
For all $a_i \in A$ and $b' \neq b_{a_i+t^*} \in B$,
\begin{align*}
\|a_i + t - b'\| 
&\geq \|b' - b_{a_i+t^*}\| - \|a_i + t - b_{a_i+t^*}\| \\
&\geq (c+1)(1+\frac{2}{m})R - (\|a_i + t^* - b_{a_i+t^*}\| + \|t - t^*\|) \\
\intertext{Assuming $\norm{t-t^*}\leq \frac{2R}{m}$, and since $\norm{a_i + t^* - b_{a_i+t^*}} \leq R$, we have}
&\geq (c+1)(\|a_i + t^* - b_{a_i+t^*}\| + \|t - t^*\|) - (\|a_i + t^* - b_{a_i+t^*}\|+ \|t - t^*\|) \\
&= c(\|a_i + t^* - b_{a_i+t^*}\| + \|t - t^*\|) \\
&\geq c\|a_i + t - b_{a_i+t^*}\|=c\|a_i + t - b_{a_i+t}\|
\end{align*}
\end{proof}

\newpage

\section{Proof of Lemma \ref{lem:convert-1d-reduction}}
\label{appx:convert-1d-reduction}
\begin{proof}
For property 1, if $x$ and $y$ are orthogonal, $A(x) \subseteq B(y)$. Therefore $\operatorname{CD}(A(x), B(y)) = 0$. For property 2, if $x$ and $y$ are orthogonal, without loss of generality, assume $t \geq 0$. Then the right most point of $A(x)$ will be shifted $t$ to the right, and it will be $t$ away from the rightmost point in $B(y)$. If $t < 1$, then because $x$ and $y$ are orthogonal, $\exists i \in [d]$ such that $x_i = y_i = 1$. Therefore between $(4 i - 4, 4 i + 1)$ in $B(y)$ are $\{4 i - 3, 4 i\}$ and in $A(x)$ are $\{4 i - 2, 4 i - 1\}$. The first point in $A(x)$ in this range for $0 \leq t < 1$ will be $1 + t$ from the first point in $B(y)$ and the second point in $A(x)$ in this range will be $1 - t$. away from the second point in $B(y)$. Therefore the Chamfer distance under translation for the points in $A(x)$ in this range alone is $2$. Therefore $\operatorname{CD}(A(x) + t, B(y)) \geq \max\{1, t\}$.

For property 3, without loss of generality assume $t \geq w$. Therefore the left most point of $A(x)$ is to the right of the right most point of $B(y)$. Therefore all points of $A(x) + t$ have the nearest point in $B(y)$ at $w$. The left most and right most point will be $t - w$ and $4 d + 1 + t - w$ away from this point. For every index $i$, the coordinate gadget will contribute $2 (4 i - 1.5 + (t - w))$ to the Chamfer distance regardless of whether $x_i = 0$ or $x_i = 1$; this is because the coordinate gadget will have two points, with center of mass at $4 i - 1.5$, regardless, and thus the center of mass will be $4 i - 1.5 + (t - w)$ away from the right most point of $B(y)$. Therefore the total Chamfer distance is

\begin{align*}
    \operatorname{CD}(A(x) + t, B(y)) &= t - w + 4 d + 1 + t - w + \sum_{i = 1}^d 2 (4 i - 1.5 + (t - w)) \\
    &= 2 t + 4 d - 2 w + 1 + 4 (d^2 + d) + d (- 3 + 2 (t - w)) \\
    &= t \cdot 2 (d + 1) + 4 d^2 + 5 d + 1 - 2 w - 2 w d \\
    &= t \cdot 2 (d + 1) + 4 d^2 + 5 d + 1 - 2 (4 d + 1) - 2 (4 d + 1) d \\
    &= t \cdot 2 (d + 1) - 4 d^2 - 5 d - 1
\end{align*}

The same argument can be used when $t < - w$, with all the points in $A(x)$ having its nearest neighbor in $B$ be the left most point in $B(y)$ at 0.
\end{proof}

\section{Proof of Lemma \ref{lem:gadget-combination}}
\label{appx:gadget-combination}
\begin{proof}
 Proof by construction, let $U = 10 n \Delta$, where $\Delta$ is the diameter of $\bigcup_{i = 1}^k A_i \cup B_i$. Then define $A = \bigcup_{i = 1}^k A_i + (U \cdot i, 0, ...)$. Likewise define $B = \bigcup_{i = 1}^k  B_i + (U \cdot i, 0, ...)$. We argue that under the optimal translation, all points in $A_i$ must now have their nearest neighbors in $B$ in $B_i$. Note that when $|t| < U / 3$, this is true. Let $|t| \geq U / 3$, and without loss of generality, let $t > 0$. Then the rightmost point (the point with greatest value in dimension 1) in $A_k$ must be at least $t - 2 \Delta$ away from rightmost point in $B$, which means the Chamfer distance is at least $t - 2 \Delta > U / 3 - \delta > 2 n \Delta$. This is guaranteed to not be optimal as for $t = 0$, $\operatorname{CD}(A + t, B) \leq n \Delta$. Therefore $t < U / 3$. A similar argument can be used to argue $t > - U / 3$, except with the left most point in $A_1$ and the left most point in $B$, and therefore $|t| < U / 3$ every point in $A_i$ has its nearest neighbor in $B$ in $B_i$.
\end{proof}

\section{Algorithm 1: Pseudocode}
\label{appx:algo-1-pseudo-code}
\begin{tcolorbox}[title=Exact ChamferDistanceUnderTranslation in \textsc{$\mathbb{R}^1$}]
\begin{algorithmic}
\State Sort $B$ in ascending order
\State Let $T \gets [\,]$
\State Let $I \gets \{\}$
\ForAll{$a \in A$}
    \For{$i \gets 1$ \textbf{to} $|B|$}
        \State Let $t \gets B[i] - a$
        \State Append $t$ to $T$
        \State Update $I[t].n_b \gets I[t].n_b + 1$
    \EndFor
    \For{$i \gets 1$ \textbf{to} $|B| - 1$}
        \State Let $t \gets \tfrac{1}{2}(B[i] + B[i+1]) - a$
        \State Append $t$ to $T$
        \State Update $I[t].n_m \gets I[t].n_m + 1$
    \EndFor
\EndFor
\State Sort $T$ in ascending order
\State Let $d \gets -|A|$
\State Let $\mathrm{cd} \gets \textsc{ChamferDistance}(\{a + T[1] \mid a \in A\},\, B)$
\State Let $\mathrm{cd}^{*} \gets \mathrm{cd}$
\State Let $t_{\mathrm{prev}} \gets T[1]$
\For{$i \gets 1$ \textbf{to} $|T|$}
    \State Let $t \gets T[i]$
    \State Update $\mathrm{cd} \gets \mathrm{cd} + (t - t_{\mathrm{prev}}) \cdot d$
    \State Update $\mathrm{cd}^{*} \gets \min(\mathrm{cd}^{*},\, \mathrm{cd})$
    \State Update $d \gets d + 2\,I[t].n_b - 2\,I[t].n_m$
    \State Update $t_{\mathrm{prev}} \gets t$
\EndFor
\State \Return $\mathrm{cd}^{*}$
\end{algorithmic}
\end{tcolorbox}

\section{Algorithm 1: Extension to Higher Dimensions}
\label{appx:algo-1-extension}
This algorithm does not extend to higher dimensions for general distance metric. The main failure is that the Hessian of the Chamfer distance function parametrized by translation is not zero between points of undifferentiability because of the interaction of the higher dimensions; for example for $\ell_2$, $\frac{\partial}{\partial x} \|(x, y)\|_2$ depends on $y$, and the dependence makes the hessian of the norm function non-zero.

The result is that Chamfer distance for a given translation is not computable naively from the distance and Jacobian for a nearby translation; and more importantly, that stationary points are not necessarily in the level set of a point of undifferentiability. Thus, we do not believe there is an easy set of candidates, which when checked gives an optimal solution. However, we do prove that the set of candidates will provide a 2 approximation, and that the optimal translation will be nearby one of the candidates.

However, for the $\ell_1$ and $\ell_\infty$ distance metrics, the Hessian is zero between points of undifferentiability; both of norms are locally linear, outside of points of undifferentiability. This allows for a $\mathcal{O}(m^{d + 1} n^{d +1})$ algorithm for computing exact Chamfer distance in $\mathbb{R}^d$. Define $A$ and $B$ as aligned in dimension $d_i$ if and only if there exists $a \in A, b \in B$ such that $a_i = b_i$. The set of candidates is the set of all translations such that $A$ and $B$ are aligned in all dimensions. One can prove this statement by fixing all but one dimension and repeating the exact same argument for Chamfer distance under translation in $\mathbb{R}^1$. Therefore for any translation in the other dimensions, this dimension must be aligned, and therefore all dimensions must be aligned.

There are $\mathcal{O}(m n)$ possible alignments in each dimension, which results in $\mathcal{O}(m^d n^d)$ possible candidates. Computing the Chamfer distance of each candidate takes $\mathcal{O}(m n)$ naively, granting a $\mathcal{O}(m^{d + 1} n^{d + 1})$ algorithm. The technique of shaving off the computation time of Chamfer distance is unused because there is not a natural ordering of the candidates and the points of undifferentiability which are local maxima, the mid point equivalents, explodes in number in higher dimensions, which necessitates avoiding iteration over all points of undifferentiability.
\section{Algorithm 2: Approximate Nearest Neighbor Data Structure}
\label{appx:nn-data-structure}
Let $L\geq \min_{b_1,b_2\in B}\frac{\norm{b_1-b_2}}{c\log(n)}$ be a lower bound for computing the nearest neighbor of any point under a candidate translation, and let $U\leq diam(A)+diam(B)$ be an upper bound. 

Initialize a simple standard hash function and place all elements in $B$ in it, which will be used to assess $R_0=0$ distances after translation. Following the standard approach for computing approximate nearest neighbor by approximate near neighbor, for $i=1,...,\lceil \log_c(U/L)\rceil$, let $R_i=c^{i-1}L$ be a scale of the LSH. For each $R_i$, we create a locality-sensitive hashing (LSH) data structure on $B$ as described in \citet{optimal-data-dependent-hashing}, for calculating $(R_i,cR_i)$ approximate near neighbor. Then for each query point $q$, we do binary search over the $\log_c(\frac{U}{L})$ scales to find the smallest scale $R_i$ at which $q$ has an neighbor within $cR_i$. Preprocessing takes time $\mathcal{O}(n^{1+\frac{1}{2c^2-1}})$ and space $\mathcal{O}(n^{1+\frac{1}{2c^2-1}}+dn)$ per scale.  Queries take time $\mathcal{O}(\log\log_c(\frac{U}{L}))$ We drop the $\log_c(\frac{U}{L})$ terms in the headline results because $U$ and $L$ are application-dependent and in many applications this term will be quite small, and because the single-log term $\log_c(\frac{U}{L})$ is only paid once during preprocessing.

In lieu of calculating the lower and upper bounds manually, one can only initialize scales when they are queried by a point during binary search. To save time and space for a lower bound, while initializing every scale one can check if there are two unique points in $B$ that map to the same bucket in the LSH and have distance $\leq cR_i$. If not, this is the lower bound and one can stop. This ensures that, for every translated $a$, there is at most one unique point in $B$ within $cR_1=cL$ of it; its distance is then computed exactly.
\section{Algorithm 2: Runtimes and Correctness}
\begin{lemma}
    Algorithm 2 variant 1 yields a $(2+\eps)$-approximation in time $\mathcal{O}(\frac{1}{\eps^3}(m+n)nd\log(n)\log(\frac{mn}{\eps}))$ with probability $\geq 0.9$
\end{lemma}
\begin{proof}
    We sample $k=\lceil \frac{24}{\eps}\rceil$ points $a_i\sim U[A]$ so that a translation $\tilde{t}=b_j-a_i$ yields $\op{CD}(A+\tilde{t},B)\leq (2+\frac{\eps}{4})\mathrm{OPT}$ with probability 
    $\geq 95$\%,
    \[Pr[\text{no sampled $\tilde{t}$ has $\op{CD}(A+\tilde{t},B)\in [\mathrm{OPT},(1+\frac{\eps}{4})]$}]\leq (1-\frac{\eps}{8})^k\quad \text{By Lemma \ref{lem:appx-chamfer-translations}}.\]
    For $k=k'(\frac{8}{\eps})$, this is $\leq e^{-k'}$. We want this to be $\leq 0.05$, which yields $k'\geq -\ln(0.05)$, so $k'=3$ suffices.
    
    For each translation tested, we run the algorithm from \citet{near-linear-time-algo-chamfer-distance} (setting the $\eps$ in that algorithm to our $\eps/16=\mathcal{O}(\eps)$) $l$ times and then take their median, scaled by $\frac{1}{1-\frac{\eps}{16}}$ so that it is not an underestimate. Each run of that algorithm gives an estimate in the range $[\op{CD}(A+t,B),(\frac{1+\frac{\eps}{16}}{1-\frac{\eps}{16}})\mathrm{CD(A+t,B)}]\subset [\op{CD}(A+t,B),({1+\frac{\eps}{4}})\mathrm{CD(A+t,B)}]$ with probability $\geq 0.99$. We set $l$ so that the probability that a given median estimate falls outside the range, call that event $F$, is $\leq \frac{0.05}{kn}\leq  \frac{0.05\eps}{48mn}$. Using the standard bound on the probability that the median fails, we have
    $\Pr[F]\leq (2\sqrt{0.01(0.99)})^l\leq (0.02)^l$
    We want this to be $\leq \frac{0.05\eps}{48mn}$, which yield $l\geq \frac{\log(\frac{0.05\eps}{12mn})}{\log(0.02)}=\mathcal{O}(\log(\frac{mn}{\eps}))$. This ensures that every translation tested yields an approximation $\hat{\op{CD}}(A+t,B)$ with $\op{CD}(A+t,B)\leq \hat{\op{CD}}(A+t,B)\leq (1+\frac{\eps}{4}){\op{CD}}(A+t,B)$ with probability $\geq 0.95$. Thus with probability $\geq 0.9$, the minimal approximate Chamfer distance among the translations yields a $(1+\frac{\eps}{4})(2+\frac{\eps}{4})\leq (2+\eps)$- approximation of the true minimum Chamfer distance.\\
    Thus there are $kln=\mathcal{O}(\frac{\log(\frac{mn}{\eps})}{\eps}n)$ translations tested, each of which takes time $\mathcal{O}(\frac{1}{\eps^2}(m+n)d)$, for a $\mathcal{O}(\frac{1}{\eps^3}(m+n)nd\log(n)\log(\frac{mn}{\eps}))$ total time algorithm.
\end{proof}
\begin{lemma}
Algorithm 2 variant 2 yields a $(2+\eps)c$-approximation of $\mathrm{OPT}$ in time $\mathcal{O}(\frac{1}{\eps}mn^{1+\frac{1}{2c^2-1}}d\log(n))$
\begin{proof}
Each $\Delta_i$ is a $c$ approximation of the distance from $a_i$ to its nearest neighbor $b_{a_i+t}$, and never an underestimate. $\sum_{i\in [m]}\Delta_i$ is therefore always a $c$ approximation of the Chamfer distance at that translation. By Lemma \ref{lem:appx-chamfer-translations}, the minimum such approximation over all n translations over $k=\lceil\frac{6}{\eps}\rceil$ points in $A$ is a $(2+\eps)c$ approximation for the Chamfer distance under translation with probability $\geq$ 90\%,
\[Pr[\text{No sampled translations give $(2+\eps)$ approximation for OPT}]\leq (1-\frac{\eps}{2})^k.\]
For $k=k'(\frac{2}{\eps})$, this is $\leq e^{-k'}$. Setting this $\leq 0.1$ gives $k'\geq -\ln(0.1)$. $k'=3$ suffices.
\\
This gives $\mathcal{O}(\frac{n}{\eps})$ total translations tested, each of which takes time $\mathcal{O}(mn^{\frac{1}{2c^2-1}}d\log(n))$ after preprocessing, for running time $\mathcal{O}(\frac{1}{\eps}mn^{1+\frac{1}{2c^2-1}}d\log(n))$. Preprocessing takes time $\mathcal{O}(n^{\frac{1}{2c^2-1}}d\log(n))$ to build the LSH, see Appendix \ref{appx:nn-data-structure}

\end{proof}

\end{lemma}

\section{Algorithm 3: Correctness}
\label{appx:proof-of-local-net-theorem}
\begin{proof}
Let \(Q\) be a uniformly random subset of \(A\) of size \(
k=\left\lceil \frac{2}{\gamma}\log\frac{1}{\delta}\right\rceil\), and \(T_{samp} =  \{b_i-a_j : a_j\in Q,\ b_i\in B\}\).
Let \(G_\gamma\subseteq A\) be the good set from Lemma~\ref{lem:closest-pair-upper-bound}, so
\[|G_\gamma|\ge \frac{\gamma}{1+\gamma}m\ge \frac{\gamma m}{2}.\]
Then,
\begin{align*}
\Pr\left[
\exists t_c \in T_{\mathrm{samp}}
:
\|t_c-t^*\|
\le
(1+\gamma)\frac{\mathrm{OPT}}{m}
\right] &\ge
\Pr[Q\cap G_\gamma\neq \emptyset] \\
&=
1-\frac{\binom{m-|G_\gamma|}{k}}{\binom{m}{k}} \\
&\ge
1-\left(1-\frac{|G_\gamma|}{m}\right)^k\\ &\ge
1-e^{-k\gamma/2}.
\end{align*}
So, with $k=\left\lceil \frac{2}{\gamma}\log\frac{1}{\delta}\right\rceil$, we get \[\Pr\left[\exists t_c \in T_{\mathrm{samp}}: \|t_c-t^*\| \le (1+\gamma)\frac{\mathrm{OPT}}{m} \right] \ge 1-\delta.\]

Let \(u:=\min_{t_c\in T_{\mathrm{samp}}}\operatorname{CD}(A+t_c,B)\). On the successful sampling event, Lemma \ref{lem:appx-chamfer-translations} gives
\[
    \mathrm{OPT}\le u\le (2+\gamma)\mathrm{OPT}.
\]
Set \(R=(1+\gamma)u/m\) and \(\rho=\eps u/(hm)\) for \(h\ge 2+\gamma\). Since the successful event gives a candidate \(t_c\in T_{\mathrm{samp}}\) with
\[
\|t_c-t^*\|\le (1+\gamma)\mathrm{OPT}/m\le R,
\]
we have \(t^*\in B(t_c,R)\), and the \(\rho\)-net $N_{t_c}$ of this ball contains some point \(s\in N_{t_c}\) with $\|s-t^*\|\le \rho.$
Therefore,
\begin{align*}
\operatorname{CD}(A+s,B)
&\le \mathrm{OPT}+m\rho \\&= \mathrm{OPT}+\frac{\eps u}{h} \\&\le \left(1+\frac{(2+\gamma)\eps}{h}\right)\mathrm{OPT} \\&\le (1+\eps)\mathrm{OPT}.
\end{align*}
Since the algorithm minimizes over all net points, it achieves cost at most that of \(s\), and hence obtains the desired \((1+\eps)\)-approximation.
It remains to bound the running time. We have
\[
|T_{\mathrm{samp}}|
=
nk
=
O\left(\frac{n}{\gamma}\log\frac1\delta\right),
\qquad
\left(\frac{R}{\rho}\right)^d
=
\left(\frac{h(1+\gamma)}{\eps}\right)^d
=
\mathcal{O}(\eps^{-d}).
\]
Thus the algorithm evaluates $O\left(\frac{n}{\gamma}\eps^{-d}\log\frac1\delta\right)$
total net points. Since each exact Chamfer-distance evaluation takes \(\mathcal{O}(mn)\) time,
the total running time is $O\left(mn^2\gamma^{-1}\eps^{-d}\log\frac1\delta\right).$
Taking \(\gamma=\eps\) and constant \(\delta\) gives
\[
\mathcal{O}(mn^2\eps^{-(d+1)}).
\]
\end{proof}
\section{Algorithm 3: Instance Sensitive Refinement}
\label{appx:algo-3-instance-sensitive-refinement}
Note that if the candidate translations are well clustered in translation space, the dependence on the number of sampled candidates can be significantly reduced. Let $T_{samp}$ be the set of sampled translations used by the local net algorithm. The basic algorithm constructs a separate local net around each candidate in $T_{samp}$, giving the number of total net points to be proportional to $|T_{samp}|\eps^{-d}$. However, it suffices to construct a single net over the union of local search balls around all candidates. 
This leaves the approximation guarantee unchanged.

In the worst case, when the balls are essentially disjoint, this recovers the original $|T_{samp}|\eps^{-d}$ bound. But when the sampled translations are closer, the union can be much smaller. For example, suppose that the sampled translations lie in a region whose $R$-covering number is $K$. Then the number of net points needed for the union of the search balls is on the order of $K\eps^{-d}$ rather than $|T_{samp}|\eps^{-d}$. This improves the bound when $K\ll|T_{samp}|$. Thus, the dependence on $|T_{samp}|$ is replaced by an instance-sensitive covering number. This can be much smaller when many sampled translations are similar, for example when many points in $A$ induce similar offsets to points in $B$. This is precisely the regime in which $A$ and $
B$ are similar up to translation: many sampled offsets $b-a$ concentrate in the same region of translation space, causing their local search balls to overlap.

\section{Algorithm 4: Runtime}
\label{appx:algo-4-runtime}
\begin{remark}
    Algorithm 4 runs in time $\mathcal{O}(mnd(n^{\frac{1}{2c^2-1}}\log(n)+\log^{3}(\frac{mR}{\eps})))=\tilde{\mathcal{O}}(mn^{1+\frac{1}{2c^2-1}}d)$ 
\end{remark}
\begin{proof}
    Building the approximate nearest neighbor data structure takes time $\mathcal{O}(n^{1+\frac{1}{2c^2-1}}\log(n)d)$. Querying it takes time $\mathcal{O}(mn^{\frac{1}{2c^2-1}}\log(n)d)$ time for each of the $\mathcal{O}(n)$ translations, for $\mathcal{O}(mn^{1+\frac{1}{2c^2-1}}\log(n)d)$ time total. For each translation, upon retrieving the $m$ nearest neighbors the median subroutine takes time $\mathcal{O}(m\log^3(\frac{mR}{\eps})d)$, along with $\mathcal{O}(md)$ time for computing the total distance to the approximate median, leading to total time $\mathcal{O}(mn\log^3(\frac{mR}{\eps})d)$ for computing the medians and total distances for all translations.
\end{proof}
\section{Algorithm 4: Further Assumption Leading to Earth Mover's Distance}
\label{appx:further-assumption-emd}
For any two point sets $A, B \subset \mathbb{R}^d$ where $|A| = m<|B|=n$, the asymmetric Earth Mover's Distance is defined as $\operatorname{EMD}(A,B) = \min_{f: A \rightarrow B\text{ injective} }\sum_{a\in A}d(a,f(a))$ \cite{a-near-linear-time-approx-emd}. Asymmetric Earth Mover's Distance under translation is analogously defined as $\operatorname{EMDuT}(A,B) = \min_{t\in \mathbb{R}^d,f:A\rightarrow B \text{ injective}}\sum_{a\in A}d(a+t,f(a))$
\begin{lemma}[Extension to Earth Mover's Distance]
Given the assumptions of Algorithm 4, if we further assume that $\forall a_i\neq a_j\in A,\norm{a_i-a_j}>R(1+\eps)$, then \[\operatorname{EMDuT}(A,B)\leq R\iff \operatorname{CDuT}(A,B)\leq R\] and \[\operatorname{EMDuT}(A,B)> R(1+\eps)\iff \operatorname{CDuT}(A,B)> R(1+\eps).\]
\begin{proof}
Suppose that $\op{CDuT} \leq R(1+\eps)$ with some optimal translation $t^*$. Then
\[
\begin{aligned}
\forall i\neq j,R(1+\eps)
&\geq \norm{a_i + t^* - b_{a_i + t^*}} + \norm{a_j + t^* - b_{a_j + t^*}} \\
&\geq \norm{(a_i + t^* - b_{a_i + t^*}) - (a_j + t^* - b_{a_j + t^*})} \\
&= \norm{(a_i - a_j) - (b_{a_i + t^*} - b_{a_j + t^*})} \\
&\geq \norm{a_i - a_j} - \norm{b_{a_i + t^*} - b_{a_j + t^*}}.
\end{aligned}
\]
\[
\implies 
\norm{b_{a_i + t^*} - b_{a_j + t^*}} \;\geq\; \norm{a_i - a_j} - R(1+\eps) \;>\; 0.
\]
Thus all $b_{a_i+t^*}$ are distinct. When all $b$ are distinct in an assignment of points from $A$ to $B$ under Chamfer distance, then this is also a valid point assignment for Earth Mover's Distance (EMD), with the same value. Thus $\op{CDuT}(A,B)\leq R(1+\eps)\implies \op{CDuT}(A,B)\geq \op{EMDuT}(A,B)$. Since Chamfer distance under translation is also a lower bound for EMD under translation, we have $\op{CDuT}(A,B)\leq R \iff \op{EMDuT}(A,B)\leq R$ and $\op{CDuT}(A,B)> R(1+\eps) \iff \op{EMDuT}(A,B)> R(1+\eps)$
\end{proof}
\end{lemma}
This means that under this additional separation assumption on $A$, algorithm 4 is a valid algorithm for the decision version of $\op{EMDuT}$ as well. Given the separation assumption on $B$, this is a natural assumption to make on $A$ as well.
\end{document}